\documentclass[12pt]{article}
\usepackage{amsmath}
\usepackage{graphicx,psfrag,epsf}
\usepackage{enumerate}
\usepackage{psfrag,epsf}
\usepackage{url} 
\usepackage{amsmath,amssymb,amsfonts,amsthm,mathtools,mathrsfs}

\usepackage{tikz}
\usepackage{pgfplots}
\usepgfplotslibrary{fillbetween}
\pgfplotsset{compat=1.16}

\usepackage{bm,bbm}
\usepackage{color}
\usepackage{subfigure}
\usepackage{algorithm,algpseudocode}
\usepackage[symbol]{footmisc}
\usepackage{scalefnt}
\usepackage{authblk} 
\usepackage{booktabs}
\usepackage{multirow,centernot}
\usepackage[colorlinks=true, citecolor=blue, urlcolor=blue]{hyperref}
\usepackage{natbib}
\usepackage{dsfont}
\usepackage[title]{appendix}
\input cyracc.def

\usepackage[left=1in,top=1.1in,right=0.5in,bottom=1in]{geometry}

\theoremstyle{definition}

\newtheorem{theorem}{Theorem}[section]

\newtheorem{example}{Example}[section]

\newtheorem{proposition}[theorem]{Proposition}
\newtheorem{remark}[theorem]{Remark}

\makeatletter
\def\@seccntformat#1{\@ifundefined{#1@cntformat}%
	{\csname the#1\endcsname\quad}
	{\csname #1@cntformat\endcsname}
}
\makeatother

\newif\ifShowComments
\ShowCommentstrue
\def\strutdepth{\dp\strutbox}
\def\druk#1{\strut\vadjust{\kern-\strutdepth
        {\vtop to \strutdepth{%
                \baselineskip\strutdepth\vss
                        \llap{\hbox{#1}\quad}\null}}}}

\title{\bf
The arithmetic-harmonic inequality index: Theory, inference, and finite-sample analysis
}

\author{
\text{Roberto Vila}$^{1}$\thanks{Corresponding author: Roberto Vila, email: {rovig161@gmail.com}
}
\,\, and
\text{Helton Saulo}$^{1,2}$ 
\\
{\small $^{1}$ Department of Statistics, University of Brasilia, Brasilia, Brazil}\\
{\small $^{2}$ Department of Economics, Federal University of Pelotas, Pelotas, Brazil}\\
}
\begin{document}
	\maketitle 	
\begin{abstract}
We investigate the arithmetic-harmonic inequality (AHI) index, a bounded and scale-invariant measure of dispersion for positive random variables, defined through the interplay between the mean and its reciprocal. We derive analytical expressions for the AHI index within the generalized inverse Gaussian (GIG) family, encompassing the inverse Gaussian and gamma distributions as important special cases. We study the associated estimator, obtain a tractable expression for its expectation, establish its asymptotic properties, and derive explicit first-order bias approximations. A Monte Carlo study is conducted to evaluate the finite-sample performance of the estimator under various scenarios. An application to GDP per capita data for countries in the Americas illustrates the role of the AHI index within the broader Atkinson family across several values of the inequality-aversion parameter. The results show the good performance of the AHI index as a tractable and interpretable measure of economic dispersion.
\end{abstract}
%


\section{Introduction}

Let $X$ be a positive random variable defined on a probability space $(\Omega,\mathcal{F},\mathbb{P})$ such that
$
\mathbb{E}[X] \in (0,\infty)
$
and
$
\mathbb{E}\left[{1}/{X}\right] \in (0,\infty)
$. We consider the population index
\begin{align}\label{AHI-index}
	J =J(X)\equiv 1 - \frac{1}{\mathbb{E}[X] \mathbb{E}[1/X]}.
\end{align}
This index depends on the interaction between the mean and the reciprocal mean. Hence, for simplicity, this index is referred to as the arithmetic-harmonic inequality (AHI) index. 

From Jensen's inequality, $0 \leqslant J < 1$.
Moreover,
$
J = 0$ if and only if $X$ is degenerate random variable;
for any $c>0$,
$J(cX) = J(X)$; and
$J(X) = J(1/X)$.

Since
$
\mathbb{E}[X]\mathbb{E}[1/X] = \mathbb{E}\left[{X_1}/{X_2}\right]
$,
where \(X_1, X_2\) are independent and identically distributed (i.i.d.) copies of \(X\), the index \(J\) admits the ratio-based representation
$
J = 1 - 1/\mathbb{E}[X_1/X_2].
$
Hence, the AHI index measures multiplicative dispersion through pairwise ratios. Specifically,
$J$ quantifies the deviation of the average ratio from unity.
Furthermore, since \(H = 1/\mathbb{E}[1/X]\) denotes the harmonic mean, \(J\) can also be expressed in terms of means as
$
J = 1 - {H}/{ \mathbb{E}[X]}
$.

\smallskip 
Using a second-order Taylor expansion of \(x \mapsto 1/x\) about \(\mathbb{E}[X]\) and taking expectations, we obtain
\[
\mathbb{E}\left[{1}/{X}\right]
\approx
\frac{1}{\mathbb{E}[X]}
+
\frac{\operatorname{Var}(X)}{\mathbb{E}^3[X]}.
\]
Substituting into the definition of \(J\) yields
\begin{align}\label{connection-CV}
	J \approx \frac{\mathrm{CV}^2}{1+\mathrm{CV}^2},
\end{align}
where \(\mathrm{CV}^2 = \operatorname{Var}(X)/\mathbb{E}^2[X]\) is the squared coefficient of variation ($\mathrm{CV}^2$) \citep{Everitt1998}.

Thus, \(J\) is a monotone transformation of \(\mathrm{CV}^2\), behaving like $\mathrm{CV}^2$ for small variability  and approaching $1$ as \(\mathrm{CV}^2 \to \infty\). 
Hence, \(J\) is a normalized, scale-free measure of dispersion bounded in \([0,1)\).


{
By choosing \(\varepsilon = 2\) in the Atkinson index \citep{Atkinson1970}
\begin{align}\label{Atkinson-index}
A(\varepsilon)
=
1 - \frac{\mathbb{E}^{1/(1-\varepsilon)}[X^{1-\varepsilon}]}{\mathbb{E}[X]},
\quad 0 \leqslant \varepsilon \neq 1,
\end{align}
one recovers the AHI index in \eqref{AHI-index}. Relation \eqref{connection-CV} then shows that, for this value, the Atkinson index is closely related to \(\mathrm{CV}^2\) while remaining bounded in \([0,1)\), yielding an interpretable and scale-free measure of inequality directly linked to a classical dispersion metric. The limiting cases \(\varepsilon \to 1\) and \(\varepsilon \to \infty\) are discussed in \cite{Assis2025}.

From a structural perspective, the Atkinson family \(A(\varepsilon)\) is monotone in \(\varepsilon\), a property that follows directly from the ordering of power means. Indeed, writing $M_r(X)=\mathbb{E}^{1/r}[X^r]$,
and letting \(r=1-\varepsilon\), it is well known that \(M_r(X)\) is strictly increasing in \(r\) whenever \(X\) is not degenerate. Hence, for \(\varepsilon_1<\varepsilon_2\), we have \(1-\varepsilon_1>1-\varepsilon_2\), implying \(M_{1-\varepsilon_1}(X)>M_{1-\varepsilon_2}(X)\), and therefore \(A(\varepsilon_1)<A(\varepsilon_2)\). This monotonicity shows that larger values of \(\varepsilon\) place progressively more weight on the lower tail of the distribution. In particular, \(\varepsilon=2\) represents a natural threshold separating regimes of moderate sensitivity to inequality (\(\varepsilon<2\)) from those dominated by the lower tail (\(\varepsilon>2\)).

As such, the AHI index \(J\) provides a balanced measure that captures disparities among lower-income observations without the excessive instability associated with higher values of \(\varepsilon\). Moreover, compared to more extreme specifications of the Atkinson index, \(J\) exhibits improved statistical stability, making it particularly suitable for empirical applications. Finally, relative to the Gini coefficient, which reflects overall inequality, \(J\) places greater, but still controlled, emphasis on the lower tail, thereby complementing standard inequality measures.
}

In this paper, we develop a comprehensive study of the AHI index \eqref{AHI-index} and its estimator. First, we derive analytical expressions for $J$ under the generalized inverse Gaussian family \citep{Seshadri1997,Perreault1999}, including the inverse Gaussian and gamma distributions as important special cases. Second, we introduce the estimator based on sample arithmetic and harmonic means and obtain an exact integral representation for its expectation using Laplace transform techniques. Third, we establish large-sample properties, including strong consistency, asymptotic normality, and a first-order bias expansion. Fourth, we conduct a Monte Carlo simulation study to investigate the finite-sample behavior of the estimator under several parameter configurations. Finally, we present an application to real GDP per capita data for countries in the Americas, which situates the AHI index within the broader Atkinson family by computing $\widehat{A}(\varepsilon)$ for several values of $\varepsilon$ and confirming the monotonicity and threshold properties established in earlier sections.

The remainder of the paper is organized as follows. Section \ref{populato_index} investigates the population properties of the index and derives closed-form expressions under the generalized inverse Gaussian family, including important special cases. Section \ref{integral_form} develops an integral representation for the expectation of the estimator and provides explicit bias expressions. Section \ref{large_sample} establishes large-sample properties, including consistency, asymptotic normality, and a first-order bias expansion. Section \ref{simulation_study} presents a Monte Carlo study to assess the finite-sample performance of the estimator under different parameter configurations. Section \ref{application} presents an application to GDP per capita data for countries in the Americas. Finally, Section \ref{conclusion} contains some concluding remarks.


\section{Special cases of the index}\label{populato_index}

In this section, we investigate the behavior of the AHI index $J$, given in \eqref{AHI-index},
at the population level for important classes of positive distributions. The goal is to obtain explicit analytical expressions for $J$ and to understand how it depends on the underlying parameters of the model. Particular emphasis is given to the generalized inverse Gaussian (GIG) family, which provides a flexible and unifying framework encompassing the inverse Gaussian and gamma distributions as special cases.
\begin{example}[Generalized inverse Gaussian distribution]
	\label{GIG}
	A random variable $X$ is said to follow a $\mathrm{GIG}(p,a,b)$ distribution if its density function is given by
	\[
	f_X(x)
	=
	\frac{(a/b)^{p/2}}{2K_p(2\sqrt{ab})} \,
	x^{p-1}
	\exp\{-ax - b/x\}, \quad x>0; \ a,b>0, p\in\mathbb{R},
	\]
	where $K_p(\cdot)$ is the modified Bessel function of the second kind, defined by
	\[
	K_p(z)
	=
	\frac{1}{2}
	\left(\frac{z}{2}\right)^p
	\int_0^\infty x^{-p-1}
	\exp\left\{-x - \frac{z^2}{4x}\right\} {\rm d}x,
	\quad 
	z>0.
	\]
	
	For any $r \in \mathbb{R}$ such that the expectation exists,
	\begin{align}\label{exp-1}
	\mathbb{E}[X^r]
	=
	\left(\frac{b}{a}\right)^{r/2}
	\frac{K_{p+r}(2\sqrt{ab})}{K_p(2\sqrt{ab})}.
	\end{align}
	
	Taking $r=1$ and $r=-1$, we obtain
	\begin{align}\label{exp-2}
	\mathbb{E}[X]
	=
	\sqrt{\frac{b}{a}} \, 
	\frac{K_{p+1}(2\sqrt{ab})}{K_p(2\sqrt{ab})},
	\quad 
	\mathbb{E}\left[{1}/{X}\right]
	=
	\sqrt{\frac{a}{b}} \, 
	\frac{K_{p-1}(2\sqrt{ab})}{K_p(2\sqrt{ab})},
\end{align}
	respectively.
	
	Then
	\begin{align}\label{index-J}
		J
		=
		1
		-
		\frac{K_p^2(2\sqrt{ab})}
		{K_{p+1}(2\sqrt{ab})\,K_{p-1}(2\sqrt{ab})}.
	\end{align}
	
	The index $J$ depends only on the product $ab$ and the parameter $p$. The structure is entirely governed by ratios of modified Bessel functions.
\end{example}

\begin{example}[Inverse Gaussian distribution]\label{ex-1}
	Consider the parametrization \(p=-{1}/{2}\), \(a={\lambda}/{(2\mu^2)}\), and \(b={\lambda}/{2}\), with $\mu,\lambda>0$. Under this choice, \(X \sim \mathrm{IG}(\mu,\lambda)\).
	
	Using the identities
	\[
	K_{1/2}(z)
	=
	\sqrt{\frac{\pi}{2z}}\,\exp\{-z\},
	\quad
	K_{3/2}(z)
	=
	K_{1/2}(z)\left(1+\frac{1}{z}\right),
	\quad
	K_{-p}(z)=K_p(z),
	\]
	and substituting \(p=-{1}/{2}\), \(a={\lambda}/{(2\mu^2)}\), and \(b={\lambda}/{2}\) into \eqref{index-J}, it follows that
	\[
	J
	=
	\frac{\mu}{\mu+\lambda}.
	\]
	
	Hence, $J$ depends only on the ratio $\mu/\lambda$ and can be interpreted as a shape parameter.
\end{example}

\begin{example}[Gamma distribution]\label{ex-2}
	%
	Consider the limit \(b \to 0^+\) with the parametrization \(p=\alpha\) and \(a=\beta\), with $\alpha,\beta>0$. In this regime, the distribution reduces to \(X \sim \mathrm{Gamma}(\alpha,\beta)\). 
	
	Furthermore, by applying the asymptotic identity
	\begin{align}\label{lim-p}
	\lim_{b\to 0^+}
	\frac{K_p^2\!\left(2\sqrt{ab}\right)}
	{K_{p+1}\!\left(2\sqrt{ab}\right)\,K_{p-1}\!\left(2\sqrt{ab}\right)}
	=
	\frac{|p|-1}{|p|},
	\quad |p|>1,
	\end{align}
	with \(p=\alpha\), and taking the limit \(b \to 0^+\) in \eqref{index-J}, it follows that
	%
	\[
	J = \frac{1}{\alpha},
	\quad \alpha>1.
	\]
	
	The index $J$ depends only on the shape parameter $\alpha$ and is independent of the rate parameter $\beta$. In particular,
	$J \to 0$ as $\alpha \to \infty$,
	indicating low dispersion, while
	$J \to 1$ as $\alpha \downarrow 1$,
	indicating high dispersion.
	
	Moreover, since for the Gamma distribution the squared coefficient of variation satisfies
	$\mathrm{CV}^2= {1}/{\alpha}$,
	it follows that
	$
	J = \mathrm{CV}^2
	$.
	Thus, the index $J$ admits a direct interpretation as a classical dispersion measure in the Gamma case.
\end{example}

\begin{remark}\label{rem-1}
Note that the GIG distribution is a very rich family that contains several important distributions as special or limiting cases. In addition to the Inverse Gaussian and Gamma distributions, it also includes other relevant models such as the Inverse Gamma (obtained as $a \to 0^+$ with $p<0$), the Reciprocal Inverse Gaussian ($p=1/2$), and the Lévy distribution (arising as a limiting case of the Inverse Gaussian when $\mu \to \infty$). These latter distributions could likewise be analyzed in order to derive the corresponding index $J$ and the bias of its estimator; however, for the sake of simplicity, we omit these cases. 
\end{remark}

\section{Integral form of the expectation of the estimator}\label{integral_form}

Given a sample $X_1,\dots,X_n$ of $X$, we define the plugging estimator of AHI index $J$ as follows
\begin{align}\label{estimator-J}
\widehat{J} = 1 - \frac{1}{\overline X\,\overline{(1/X)}},
\end{align}
where
\[
\overline X = \frac{1}{n}\sum_{i=1}^n X_i,
\quad
\overline{(1/X)} = \frac{1}{n}\sum_{i=1}^n \frac{1}{X_i},
\]
denote the mean and harmonic samples, respectively.

It is clear that
$
0 \leqslant \widehat{J} < 1;
$
$\widehat{J} = 0$ if and only if $X_1=\cdots=X_n$; for any $c>0$,
$
\widehat{J}(cX_1,\dots,cX_n) = \widehat{J}(X_1,\dots,X_n)
$;
$\widehat{J}$ is invariant under permutations of the sample; 
$
\widehat{J}(X_1,\dots,X_n) = \widehat{J}\left({1}/{X_1},\dots,{1}/{X_n}\right)
$; and that
$\widehat{J}$ admits the ratio representation
\begin{align}\label{def-J}
	\widehat{J}
	=
	1
	-
	{n^2\over \sum_{i,j=1}^n {X_i}/{X_j}}.
\end{align}
{
\begin{remark}
	The monotonicity of the Atkinson index $A(\varepsilon)$ in \eqref{Atkinson-index} with respect to the inequality aversion parameter \(\varepsilon\) extends naturally to its empirical counterpart. Indeed, for any non-degenerate sample \(X_1,\dots,X_n\) of $X$, the corresponding plug-in estimator
	\[
	\widehat A(\varepsilon)
	=
	1 - \frac{\displaystyle\left(\frac{1}{n}\sum_{i=1}^n X_i^{1-\varepsilon}\right)^{\frac{1}{1-\varepsilon}}}{\overline X}
	\]
	is strictly increasing in \(\varepsilon\). This follows from the fact that the empirical power mean
	$
	M_r(X_1,\ldots,X_n) = (\sum_{i=1}^n X_i^r/n)^{1/r}
	$
	is strictly increasing in \(r\), and the transformation \(r=1-\varepsilon\) reverses the ordering. Consequently, both the population index \(A(\varepsilon)\) and its plug-in estimator \(\widehat A(\varepsilon)\) exhibit the same monotone behavior.
	
	In particular, since \(J = A(2)\), we obtain the ordering
	\[
	\varepsilon < 2 \;\Rightarrow\; A(\varepsilon) < J,
	\quad
	\varepsilon > 2 \;\Rightarrow\; A(\varepsilon) > J,
	\]
	and similarly at the sample level,
	\[
	\varepsilon < 2 \;\Rightarrow\; \widehat A(\varepsilon) < \widehat J,
	\quad
	\varepsilon > 2 \;\Rightarrow\; \widehat A(\varepsilon) > \widehat J,
	\]
	where $\widehat J$ is as given in \eqref{estimator-J}.
	Thus, \(J\) and its plug-in estimator \(\widehat J\) act as natural reference points within the Atkinson family, both at the population and empirical levels, separating regimes of lower and higher sensitivity to the lower tail of the distribution.
\end{remark}
}

\newpage

\begin{theorem}\label{integral representation}
Let  \(X_1,\dots,X_n\) be a non-degenerate sample of $X$.
An integral representation for the expectation of $\widehat{J}$ is given by
\begin{align*}
	\mathbb{E}[\widehat{J}]
	=
	1 
	- 
	n^2 
	\int_0^\infty \int_0^\infty
	\mathcal{L}^n_{X,1/X}(s,t)
	{\rm d}s {\rm d}t,
\end{align*}
where  $\mathcal{L}_{X,Y}(s,t)=\mathbb{E}[\exp\{-sX-tY\}]$ denotes the joint Laplace transform of $(X,Y)^\top$.
\end{theorem}
\begin{proof}
Using the identity
\[
\frac{1}{ab}
= \int_0^\infty \int_0^\infty \exp\{-sa - tb\}\,{\rm d}s {\rm d}t,
\quad 
a,b>0,
\]
together with the representation of $\widehat{J}$ given in \eqref{def-J}, Tonelli's theorem yields
\begin{align*}
	\mathbb{E}[\widehat J]
	&= 
	1 
	- 
	n^2 
	\int_0^\infty \int_0^\infty
	\mathbb{E}\left[
	\exp\left\{-s\sum_{i=1}^n X_i - t\sum_{i=1}^n (1/X_i)\right\}
	\right] {\rm d}s {\rm d}t
	\\[0,2cm]
	&=
	1 
	- 
	n^2 
	\int_0^\infty \int_0^\infty
	\mathbb{E}^n\left[\exp\left\{-sX - t/X\right\}\right]
	{\rm d}s {\rm d}t,
\end{align*}
where the last equality follows from the fact that the sample $X_1,\dots,X_n$ is i.i.d. with the same distribution as $X$.  This completes the proof.
\end{proof}

\begin{example}[Generalized inverse Gaussian distribution]
	Let \(X \sim \mathrm{GIG}(p,a,b)\). Using the Bessel integral identity
	\[
	\int_0^\infty x^{p-1} \exp\{-Ax - B/x\}\,{\rm d}x
	=
	2\left(\frac{B}{A}\right)^{p/2}
	K_p\left(2\sqrt{AB}\right),
	\quad A,B>0,
	\]
	with the substitutions \(A=a+c\) and \(B=b+d\), we obtain
	\[
	\int_0^\infty x^{p-1} \exp\{-(a+c)x - (b+d)/x\}\,{\rm d}x
	=
	2\left(\frac{b+d}{a+c}\right)^{p/2}
	K_p\left(2\sqrt{(a+c)(b+d)}\right).
	\]
	
	It follows that the joint Laplace transform of \((X,1/X)^\top\) is given by
	\[
	\mathcal{L}_{X,1/X}(c,d)
	=
	\left(\frac{a}{a+c}\right)^{p/2}
	\left(\frac{b+d}{b}\right)^{p/2}
	\frac{K_p\left(2\sqrt{(a+c)(b+d)}\right)}
	{K_p\left(2\sqrt{ab}\right)}.
	\]
	
	Consequently, from Theorem \ref{integral representation}, upon making the change of variables 
	$x=\sqrt{(a+s)(b+t)}$ and $y=\sqrt{{(b+t)}/{(a+s)}}$, we deduce
	%
	%
	%
	%
	%
	%
	\[
	\mathbb{E}[\widehat{J}]
	=
	1
	-
	\frac{2n^2 \left({a}/{b}\right)^{np/2}}{K_p^n\left(2\sqrt{ab}\right)}
	\int_{\sqrt{ab}}^\infty
	x K_p^n(2x)
	\left(
	\int_{b/x}^{x/a} y^{np-1}\,{\rm d}y
	\right)
	{\rm d}x.
	\]
	
	If \(np \neq 0\), the inner integral evaluates to
	\[
	\int_{b/x}^{x/a} y^{np-1}\,{\rm d}y
	=
	\frac{1}{np}
	\left[
	\left(\frac{x}{a}\right)^{np}
	-
	\left(\frac{b}{x}\right)^{np}
	\right].
	\]
	
	Consequently,
	\begin{align}\label{e-J}
		\mathbb{E}[\widehat{J}]
		=
		1
		-
		\frac{2n \left({a}/{b}\right)^{np/2}}{p K_p^n\left(2\sqrt{ab}\right)} \,
		\int_{\sqrt{ab}}^\infty
		x K_p^n(2x)
		\left[
		\left(\frac{x}{a}\right)^{np}
		-
		\left(\frac{b}{x}\right)^{np}
		\right]
		{\rm d}x.
	\end{align}
	For \(n \geqslant 2\), the integral above does not admit a closed-form expression in general.
	
	By using \eqref{e-J} and Example \ref{GIG} it follows that the bias of $\widehat{J}$ with respect to $J$ can be expressed as 
	\[
	\mathrm{Bias}(\widehat{J},J)
	=
	-
	\frac{2n \left({a}/{b}\right)^{np/2}}{p K_p^n\left(2\sqrt{ab}\right)} \,
	\int_{\sqrt{ab}}^\infty
	x K_p^n(2x)
	\left[
	\left(\frac{x}{a}\right)^{np}
	-
	\left(\frac{b}{x}\right)^{np}
	\right]
	{\rm d}x
	+
	\frac{K_p^2(2\sqrt{ab})}
	{K_{p+1}(2\sqrt{ab})\,K_{p-1}(2\sqrt{ab})}.
	\]
\end{example}

\begin{example}[Inverse Gaussian distribution]
	In this case, setting \(p=-{1}/{2}\), \(a={\lambda}/{(2\mu^2)}\), and \(b={\lambda}/{2}\) in \eqref{e-J}, we obtain
	\begin{align*}
		\mathbb{E}[\widehat{J}]
		&=
		1
		+
		\frac{4n \mu^{n/2}}{ K_{1/2}^n\!\left(\lambda/\mu\right)} 
		\int_{\lambda/(2\mu)}^\infty
		x \, K_{1/2}^n(2x)
		\left[
		\left(\frac{x}{\lambda/(2\mu^2)}\right)^{-n/2}
		-
		\left(\frac{\lambda/2}{x}\right)^{-n/2}
		\right]
		\,{\rm d}x.
	\end{align*}
	
	Using the identities
	\[
	K_{1/2}(z)
	=
	\sqrt{\frac{\pi}{2z}}\,\exp\{-z\},
	\quad
	K_{-p}(z)=K_p(z),
	\]
	it follows that
	\begin{align*}
		\mathbb{E}[\widehat{J}]
		&=
		1
		+
		\lambda^{n/2} 2^{2-n/2} n \exp\left\{n\lambda\over\mu\right\}
		\int_{\lambda/(2\mu)}^\infty
		\exp\{-2nx\}
		\left[
		\left(\frac{\lambda}{2\mu^2}\right)^{n/2} x^{1-n}
		-
		\left(\frac{2}{\lambda}\right)^{n/2} x
		\right]
		\,{\rm d}x.
	\end{align*}
	
	Using the identity
	$
	\int_c^\infty x^{\alpha-1} \exp\{-\beta x\}\,{\rm d}x
	=
	\beta^{-\alpha}\,\Gamma(\alpha,\beta c),
	$
	we deduce
	\begin{align}\label{e-J-1}
		\mathbb{E}[\widehat{J}]
		=
		1
		+
		\lambda^{n/2} 2^{2-n/2} n \exp\left\{{n\lambda\over\mu}\right\}
		\left[
		\left(\frac{\lambda}{2\mu^2}\right)^{n/2}
		(2n)^{n-2}
		\Gamma\left(2-n,\frac{n\lambda}{\mu}\right)
		-
		\left(\frac{2}{\lambda}\right)^{n/2}
		\frac{1}{4n^2} \,
		\Gamma\left(2,\frac{n\lambda}{\mu}\right)
		\right],
	\end{align}
	where \(\Gamma(a,x)\) denotes the upper incomplete gamma function.
	
	Combining \eqref{e-J-1} with Example \ref{ex-1}, it follows that the bias of \(\widehat{J}\) relative to \(J\) can be written as
	\begin{align}\label{bias-IG}
		&\mathrm{Bias}(\widehat{J},J) \nonumber
		\\[0,2cm]
		&=
		\lambda^{n/2} 2^{2-n/2} n \exp\left\{{n\lambda\over\mu}\right\}
		\left[
		\left(\frac{\lambda}{2\mu^2}\right)^{n/2}
		(2n)^{n-2}
		\Gamma\left(2-n,\frac{n\lambda}{\mu}\right)
		-
		\left(\frac{2}{\lambda}\right)^{n/2}
		\frac{1}{4n^2} \,
		\Gamma\left(2,\frac{n\lambda}{\mu}\right)
		\right]
		+
		\frac{\lambda}{\mu+\lambda}.
	\end{align}
\end{example}

\begin{remark}
Since the Lévy distribution arises as a limiting case of the Inverse Gaussian when $\mu \to \infty$ (see Remark \ref{rem-1}), it follows, by taking the limit $\mu \to \infty$ in \eqref{bias-IG}, that the estimator $\widehat{J}$ is unbiased in this case.
\end{remark}

\begin{example}[Gamma distribution]
	In this case, setting \(b \to 0^+\) in \eqref{e-J} and using the limit
	\[
	\lim_{b\to 0^+}
	\frac{2n \left({a}/{b}\right)^{np/2}}{p K_p^n\left(2\sqrt{ab}\right)}
	=
	\frac{2^{n+1} n}{p \Gamma^n(p)} \, a^{np}, \quad  p>0,
	\]
	with  \(p=\alpha\) and \(a=\beta\),
	we obtain
	\begin{align*}
		\mathbb{E}[\widehat{J}]
		=
		1
		-
		\frac{2^{n+1} n}{\alpha \Gamma^n(\alpha)} \,
		\int_{0}^\infty
		x^{1+n\alpha}
		K_\alpha^n(2x) \, 
		{\rm d}x.
	\end{align*}
	Making the change of variables $y=2x$, we can write
	\begin{align}\label{e-J-2} 
		\mathbb{E}[\widehat{J}]
		=
		1
		-
		\frac{2^{n(1-\alpha)-1} n}{\alpha \Gamma^n(\alpha)} \,
		\int_0^\infty
		y^{1+n\alpha}
		K_\alpha^n(y)\, 
		{\rm d}y.
	\end{align}
	For general $n$, this integral does not reduce to elementary Gamma functions, but it admits a standard Mellin-Barnes (Meijer-G) representation.
	
	From \eqref{e-J-2} together with Example \ref{ex-2}, the bias of \(\widehat{J}\) with respect to \(J\) can be expressed as
	\begin{align*}
		\mathrm{Bias}(\widehat{J},J)
		=
		-
		\frac{2^{n(1-\alpha)-1} n}{\alpha \Gamma^n(\alpha)} \,
		\int_0^\infty
		y^{1+n\alpha}
		K_\alpha^n(y)\, 
		{\rm d}y
		+
		{\alpha-1\over\alpha},
		\quad \alpha>1.
	\end{align*}
\end{example}

\section{Large-sample properties of the estimator}\label{large_sample}

In what follows, we present several asymptotic properties of $\widehat{J}$. In particular, we show that $\widehat{J}$ is consistent and asymptotically normal under mild moment conditions. Moreover, its asymptotic variance depends on the second-order moments of $X$ and $1/X$.
\begin{proposition}[Strong consistency]
	Let $X_1,\dots,X_n$ be i.i.d. copies of a positive random variable $X$ such that
	$
	\mu=\mathbb{E}[X] < \infty$, 
	$\nu=\mathbb{E}[1/X] < \infty$.
	We have
	\[
	\widehat{J} \xrightarrow{\rm a.s.} J.
	\]
\end{proposition}
\begin{proof}
	By the law of large numbers,
	$
	\overline X \xrightarrow{\rm a.s.} \mu
	$,
	$
	\overline{(1/X)} \xrightarrow{\rm a.s.} \nu.
	$
	Hence $\overline X \, \overline{(1/X)} \xrightarrow{\rm a.s.} \mu\nu$, and by continuity of  function $t\mapsto 1-{1}/{t}$, $t>0$, the proof follows.
\end{proof}

\begin{theorem}[Asymptotic normality]\label{Asymptotic normality}
	Let $X_1,\dots,X_n$ be i.i.d. copies of a positive random variable $X$ such that
	$
	\mathbb{E}[X^2] < \infty$, 
	$\mathbb{E}[1/X^2] < \infty$.
	We have
	\[
	\sqrt{n}(\widehat{J} - J)
	\;\xrightarrow{\mathscr{D}}\;
	N(0,\sigma_J^2),
	\]
	where $\xrightarrow{\mathscr{D}}$ denotes convergence in distribution, and
	\[
	\sigma_J^2
	\equiv 
	\frac{\mathrm{Var}(X)}{\mu^4 \nu^2}
	+ \frac{\mathrm{Var}(1/X)}{\mu^2 \nu^4}
	+ \frac{2\,\mathrm{Cov}(X,1/X)}{\mu^3 \nu^3},
	\quad 
	\mu = \mathbb{E}[X], \quad \nu = \mathbb{E}[1/X]. 
	\]
\end{theorem}
\begin{proof}
	Let $\boldsymbol{Z}=(\overline X, \overline{(1/X)})^\top$. By the multivariate central limit theorem,
	$
	\sqrt{n}\left\{
	\boldsymbol{Z}
	- 
	(\mu, \nu)^\top
	\right\}
	\;\xrightarrow{\mathscr{D}}\;
	N(0,\bf\Sigma)
	$,
	where
	$\bf \Sigma$ is the covariance matrix of vector $(X,1/X)^\top$, given by
	\begin{align*}
		\bf \Sigma =
		\begin{pmatrix}
			\mathrm{Var}(X) & \mathrm{Cov}(X,1/X)
			\\[0,2cm]
			\mathrm{Cov}(X,1/X) & \mathrm{Var}(1/X)
		\end{pmatrix}.
	\end{align*}
	
	Let
	$g(a,b) = 1 - {1}/{(ab)}$. 
	Applying the delta method,
	\[
	\sqrt{n}(J_n - J)
	=
	\sqrt{n}\left\{g(\boldsymbol{Z}) - g(\mu,\nu)\right\}
	\;\xrightarrow{\mathscr{D}}\;
	N(0,[\nabla g(\mu,\nu)]^\top {\bf \Sigma} \, \nabla g(\mu,\nu)),
	\]
	where
	$
	\nabla g(a,b)
	= \left({1}/{(a^2 b)}, {1}/{(a b^2)}\right)^\top
	$ is the gradient vector,
	the proof follows.
\end{proof}

\begin{remark}
	Note that the asymptotic variance $\sigma_J^2$ in Theorem \ref{Asymptotic normality} depends on the variability of $X$, the variability of $1/X$, anf the interaction between $X$ and $1/X$.
	The covariance term captures the dependence between levels and reciprocals, which plays a central role in the behavior of the index.
\end{remark}

%
%
%
%
%
%

\begin{theorem}[Asymptotic bias] \label{Asymptotic bias}
	Under the conditions and notations of Theorem \ref{Asymptotic normality}, we have the expansion
	\[
	\mathbb{E}[\widehat{J}]
	= J 
	-
	\frac{1}{n}
	\left[
	\frac{\mathrm{Var}(X)}{\mu^3 \nu}
	+ \frac{\mathrm{Var}(1/X)}{\mu \nu^3}
	+ \frac{\mathrm{Cov}(X,1/X)}{\mu^2 \nu^2}
	\right]
	+ 
	\operatorname{o}\!\left(\frac{1}{n}\right).
	\]
\end{theorem}
\begin{proof}
	Let $\boldsymbol{Z}_* = (\overline X - \mu,\ \overline{(1/X)} - \nu)^\top$, where $\mu = \mathbb{E}[X]$ and $\nu = \mathbb{E}[1/X]$, denote the centered version of $\boldsymbol{Z} = (\overline X,\ \overline{(1/X)})^\top$. Since $g(a,b) = 1 - {1}/{(ab)}$, a second-order Taylor expansion of $g(\overline X,\overline{(1/X)})$ about $(\mu,\nu)$ yields
	\begin{align}\label{id-exp}
		g(\overline X,\overline{(1/X)})
		= 
		g(\mu,\nu)
		+ 
		[\nabla g(\mu,\nu)]^\top
		\boldsymbol{Z}_*
		+ 
		\frac{1}{2} \, 
		\boldsymbol{Z}_*^\top
		H(\mu,\nu)
		\boldsymbol{Z}_*
		+ 
		\operatorname{o}\!\left(\frac{1}{n}\right),
	\end{align}
	where 	
	$
	\nabla g(a,b)
	= \left({1}/{(a^2 b)}, {1}/{(a b^2)}\right)^\top
	$ and
	\[
	H(a,b)=
	\begin{pmatrix}
		-\frac{2}{a^3 b} & -\frac{1}{a^2 b^2}
		\\[0,2cm]
		-\frac{1}{a^2 b^2} & -\frac{2}{a b^3}
	\end{pmatrix}
	\]
	is the Hessian matrix.
	
	Since
	$\widehat{J}= g(\overline X,\overline{(1/X)})$
	and
	$\mathbb{E}[\boldsymbol{Z}_*] = (0,0)^\top$, taking expectations in \eqref{id-exp} gives
	\[
	\mathbb{E}[\widehat{J}]
	= 
	g(\mu,\nu)
	+ 
	\frac{1}{2} \,
	\mathbb{E}\left[
	\mathbf{Z}_*^\top H(\mu,\nu)\mathbf{Z}_*
	\right]
	+ 
	\operatorname{o}\!\left(\frac{1}{n}\right).
	\]
	
	Combining the identities
	$
	\mathbb{E}[\mathbf{Z}_* \mathbf{Z}_*^\top] = {\bf\Sigma}/n
	$ 
	and
	$
	\mathbb{E}[\boldsymbol{Z}_*^\top H(\mu,\nu)\boldsymbol{Z}_*]
	= \operatorname{tr}\!\left(H(\mu,\nu)\,\mathbb{E}[\boldsymbol{Z}_*\boldsymbol{Z}_*^\top]\right),
	$
	we get
	$\mathbb{E}[\boldsymbol{Z}_*^\top H(\mu,\nu)\boldsymbol{Z}_*]=\operatorname{tr}\!\big(H(\mu,\nu){\bf \Sigma}\big)/n$. Hence, since  $J = g(\mu,\nu)$, the above identity can be written as
	\[
	\mathbb{E}[\widehat{J}]
	= 
	J
	+ 
	\frac{1}{2n}\operatorname{tr}\!\big(H(\mu,\nu){\bf \Sigma}\big)
	+ 
	\operatorname{o}\!\left(\frac{1}{n}\right).
	\]
	Finally, multiplying the Hessian matrix by the covariance matrix completes the proof.
\end{proof}
	
%
%
%
%
%
%
\begin{example}[Generalized inverse Gaussian distribution]
	Let \(X \sim \mathrm{GIG}(p,a,b)\). Combining the identities in \eqref{exp-1} and \eqref{exp-2}, we obtain
	\begin{align*}
	&\mathbb{E}[X]
=
\sqrt{\frac{b}{a}} \, 
\frac{K_{p+1}(2\sqrt{ab})}{K_p(2\sqrt{ab})},
\ 
\mathbb{E}\left[{1}/{X}\right]
=
\sqrt{\frac{a}{b}} \, 
\frac{K_{p-1}(2\sqrt{ab})}{K_p(2\sqrt{ab})},
\
	\mathrm{Var}(X)
	=
	\frac{b}{a}
\left[
\frac{K_{p+2}(2\sqrt{ab})}{K_p(2\sqrt{ab})}
-
\frac{K_{p+1}^2(2\sqrt{ab})}{K_p^2(2\sqrt{ab})}
\right],
\\[0,2cm]
	& 	
	\mathrm{Var}\left({1}/{X}\right)
	=
	\frac{a}{b}
	\left[
	\frac{K_{p-2}\left(2\sqrt{ab}\right)}{K_p\!\left(2\sqrt{ab}\right)}
	-
	\frac{K^2_{p-1}\left(2\sqrt{ab}\right)}{K_p^2\left(2\sqrt{ab}\right)}
	\right],
	\quad 
	\mathrm{Cov}(X,1/X) = 1
	-
	\frac{K_{p+1}(2\sqrt{ab})\,K_{p-1}(2\sqrt{ab})}
	{K_p^2(2\sqrt{ab})}.
	\end{align*}
	
	Theorem \ref{Asymptotic bias} and Example \ref{GIG} imply
\begin{align}\label{e-j}
	&\mathbb{E}[\widehat{J}]
	=
		1
-
\frac{K_p^2(2\sqrt{ab})}
{K_{p+1}(2\sqrt{ab})\,K_{p-1}(2\sqrt{ab})}
\nonumber
\\[0,2cm]
	&\!-\!
	\frac{1}{n}
	\frac{K_p^2(2\sqrt{ab})}{K_{p+1}(2\sqrt{ab})\,K_{p-1}(2\sqrt{ab})} \!
	\left[
	\frac{K_{p+2}(2\sqrt{ab})}{K_{p+1}^2(2\sqrt{ab})}
	\!+\!
	\frac{K_{p-2}(2\sqrt{ab})}{K_{p-1}^2(2\sqrt{ab})}
	\!+\!
	\frac{K_p^2(2\sqrt{ab})}{K_{p+1}(2\sqrt{ab})\,K_{p-1}(2\sqrt{ab})}
	-\!
	3
	\right]
	\!+\!
	\operatorname{o}\left(\frac{1}{n}\right).
\end{align}
\end{example}
\begin{example}[Inverse Gaussian distribution]
	Considering the parametrization \(p=-{1}/{2}\), \(a={\lambda}/{(2\mu^2)}\) and \(b={\lambda}/{2}\) in \eqref{e-j} together with the identities
\begin{align*}
	&K_{1/2}(z)
	=
	\sqrt{\frac{\pi}{2z}}\,\exp\{-z\},
	\quad
	K_{3/2}(z)
	=
	K_{1/2}(z)\left(1+\frac{1}{z}\right),
	\\[0,2cm]
&K_{5/2}(z)=K_{1/2}(z)\left(1+\frac{3}{z}+\frac{3}{z^2}\right),
	\quad
	K_{-p}(z)=K_p(z),
\end{align*}
	we obtain
%
%
%
\[
\mathbb{E}[\widehat{J}]
	= 
	\frac{\mu}{\mu+\lambda}
	- 
	\frac{1}{n}
	\,
	{\mu[(\mu+\lambda)^2+\mu\lambda]\over (\mu+\lambda)^3}
	+ 
	\operatorname{o}\left(\frac{1}{n}\right).
\]
\end{example}

\begin{example}[Gamma distribution]
	In this case, setting \(b \to 0^+\) in \eqref{e-j} and using \eqref{lim-p} together the limit 
\begin{multline*}
\lim_{b \to 0^+}
\frac{K_p^2(2\sqrt{ab})}{K_{p+1}(2\sqrt{ab})\,K_{p-1}(2\sqrt{ab})}
\left[
\frac{K_{p+2}(2\sqrt{ab})}{K_{p+1}^2(2\sqrt{ab})}
+
\frac{K_{p-2}(2\sqrt{ab})}{K_{p-1}^2(2\sqrt{ab})}
+
\frac{K_p^2(2\sqrt{ab})}{K_{p+1}(2\sqrt{ab})\,K_{p-1}(2\sqrt{ab})}
-3
\right]
\\[0,2cm]
=
\frac{p-1}{p(p-2)},
\quad p>2,
\end{multline*}
with  \(p=\alpha\) and \(a=\beta\),
we obtain
%
\[
\mathbb{E}[\widehat{J}]
	= 
	{1\over\alpha}
	-
	\frac{1}{n}\, \frac{\alpha-1}{\alpha(\alpha-2)}
	+ 
	\operatorname{o}\left(\frac{1}{n}\right),
	\quad 
	\alpha>2.
\]
\end{example}
%


\section{Simulation study}\label{simulation_study}

We conduct a Monte Carlo simulation study to assess the finite-sample behavior of the estimator of the arithmetic-harmonic inequality index, given by $\widehat J = 1 - 1/[\overline X\,\overline{(1/X)}]$, under the inverse Gaussian and gamma models. For each parameter configuration and each sample size $n$, we generate $R$ independent samples, compute the estimator $\widehat J$, and approximate its finite-sample bias and mean squared error by Monte Carlo averages. More precisely, if $J(\theta)$ denotes the true value of the index under parameter vector $\theta$, then the empirical bias and mean squared error are computed as $\widehat{\mathrm{Bias}}_{MC}(\widehat J) = R^{-1}\sum_{r=1}^R [\widehat J^{(r)} - J(\theta)]$ and $\widehat{\mathrm{MSE}}_{MC}(\widehat J) = R^{-1}\sum_{r=1}^R [\widehat J^{(r)} - J(\theta)]^2$, where $\widehat J^{(r)}$ is the estimator obtained from the $r$-th replication.

In the inverse Gaussian case, we assume $X \sim \mathrm{IG}(\mu,\lambda)$, with $\mu > 0$ and $\lambda > 0$, for which the population index is $J = \mu/(\mu+\lambda)$. In the gamma case, we assume $X \sim {\rm Gamma}(\alpha,\beta)$, with $\alpha > 2$ and $\beta > 0$, for which $J = 1/\alpha$. For the inverse Gaussian distribution, we consider the sample sizes $n \in \{20,50,75,100,200\}$ and the parameter configurations $(\mu,\lambda) \in \{(0.5,1),(1,1),(2,1),(5,1)\}$. These choices generate distinct values of the true index $J = \mu/(\mu+\lambda)$ and allow us to examine how the finite-sample performance changes with the degree of dispersion. For the gamma distribution, we consider the same sample sizes and the scenarios $(\alpha,\beta) \in \{(2.5,1),(3,1),(5,1),(10,1)\}$. Since in this model the index depends only on $\alpha$, these configurations represent different levels of inequality while keeping the rate parameter fixed. In all cases, the number of Monte Carlo replications is set to $R=4000$.

Figures~\ref{fig:ig_original} and \ref{fig:gamma_original} summarize the empirical bias and mean squared error of the estimator of the AHI index as functions of the sample size. As expected, the bias and mean squared error decrease with $n$ in all scenarios, thereby supporting the asymptotic results derived in Section~\ref{large_sample}.

\begin{figure}[!ht]
\centering
\includegraphics[width=.99\textwidth]{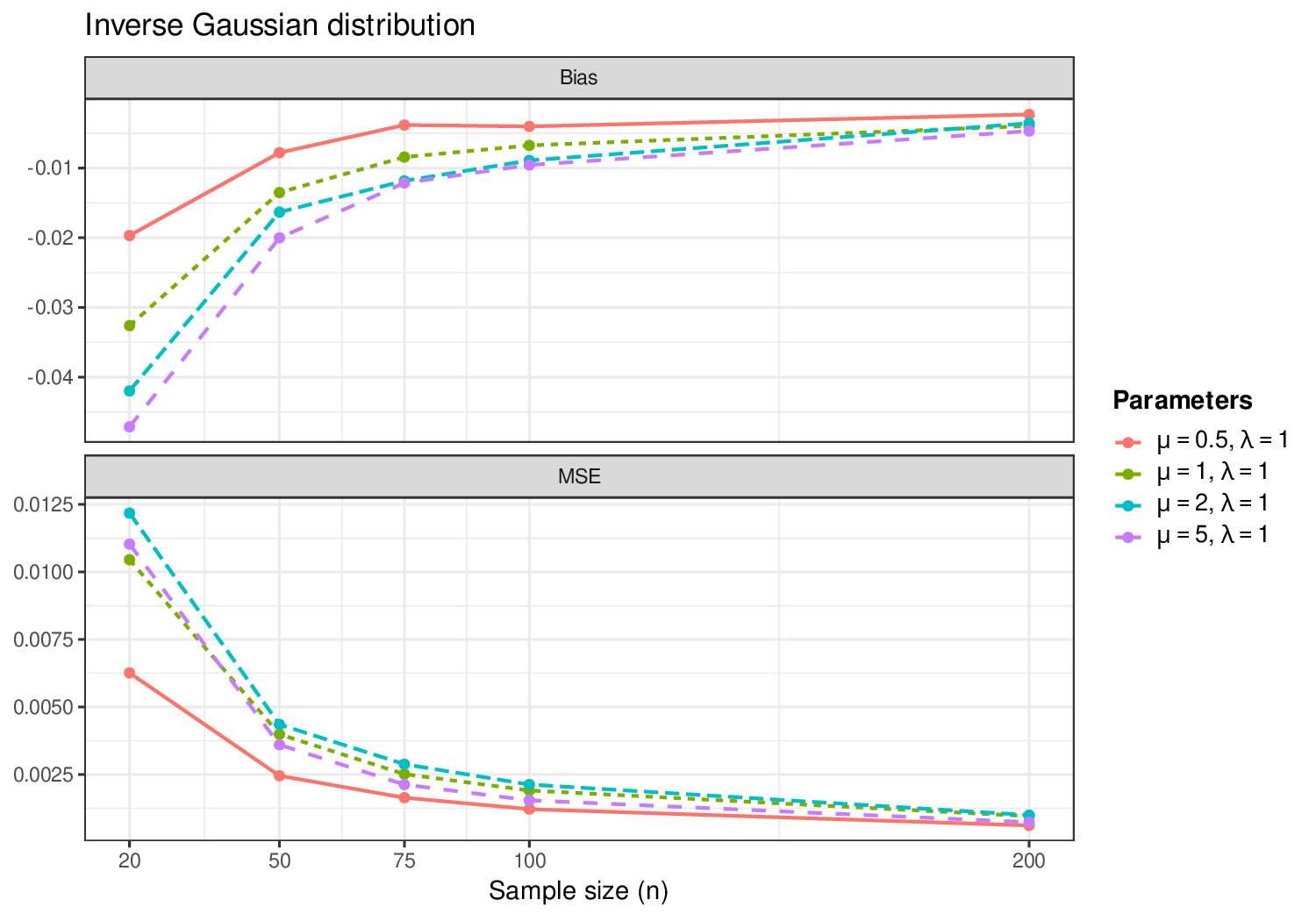}
\caption{Finite-sample bias and mean squared error of the  estimator of the AHI index under the inverse Gaussian distribution, for the parameter configurations $(\mu,\lambda) \in \{(0.5,1),(1,1),(2,1),(5,1)\}$ and sample sizes $n \in \{20,50,75,100,200\}$.}
\label{fig:ig_original}
\end{figure}

\begin{figure}[!ht]
\centering
\includegraphics[width=.99\textwidth]{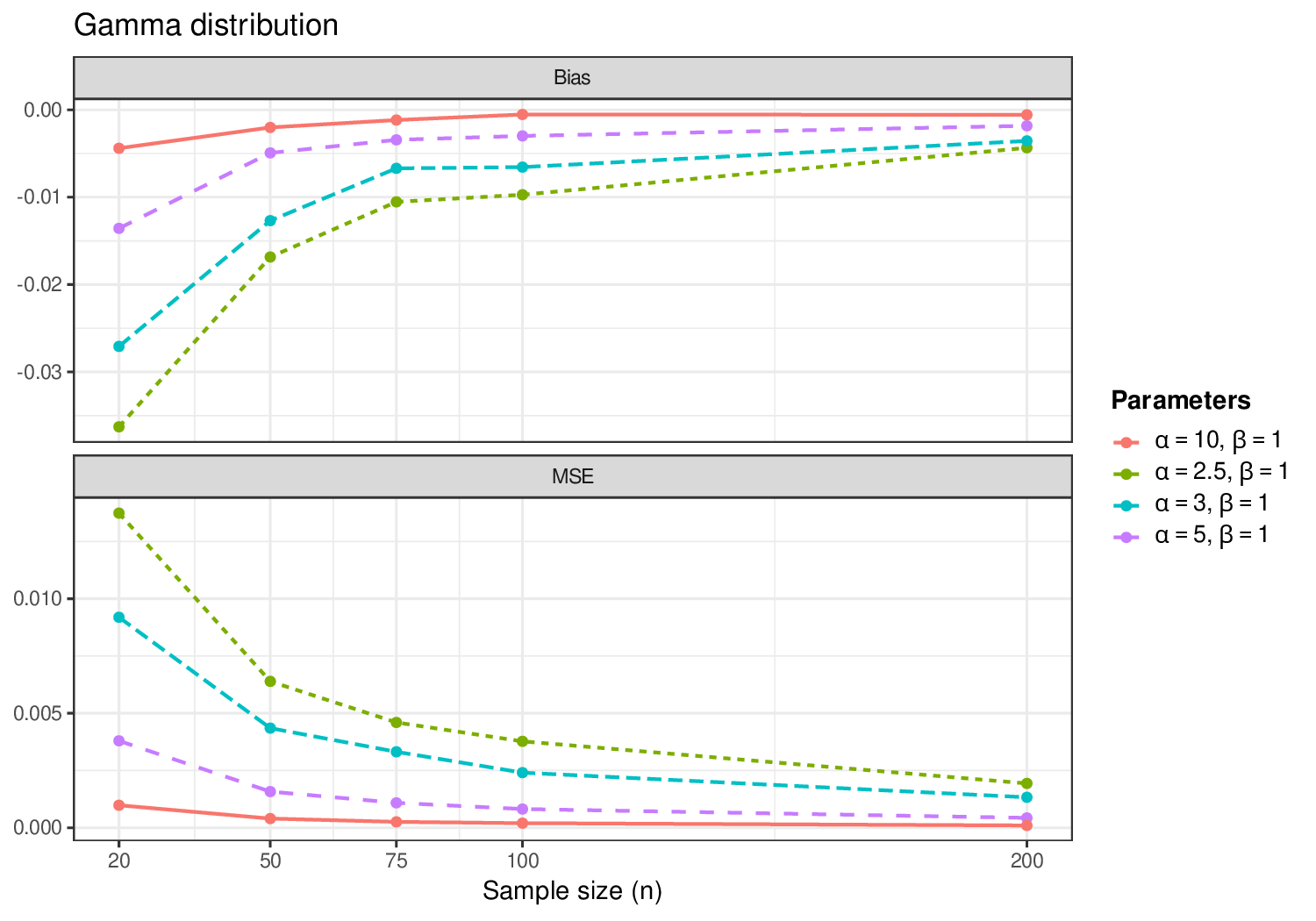}
\caption{Finite-sample bias and mean squared error of the estimator of the AHI index under the gamma distribution, for the parameter configurations $(\alpha,\beta) \in \{(2.5,1),(3,1),(5,1),(10,1)\}$ and sample sizes $n \in \{20,50,75,100,200\}$.}
\label{fig:gamma_original}
\end{figure}

\subsection{Comparison with the Atkinson index}

To place the AHI estimator within the broader Atkinson family, we compare its finite-sample bias and mean squared error against those of the plug-in Atkinson estimator
\[
\widehat{A}(\varepsilon)
=
1 - \frac{\displaystyle\left(\frac{1}{n}\sum_{i=1}^n X_i^{1-\varepsilon}\right)^{\!\frac{1}{1-\varepsilon}}}{\overline{X}}, \quad \varepsilon \neq 1,
\]
with the limiting formula $\widehat{A}(1) = 1 - \exp\!\bigl(n^{-1}\sum_{i=1}^n \log X_i\bigr)/\overline{X}$ as $\varepsilon \to 1$, for selected values $\varepsilon \in \{0.5, 1, 1.5, 2, 3, 5\}$. Since $\widehat{J} = \widehat{A}(2)$, the curve at $\varepsilon = 2$ coincides exactly with the AHI estimator. The true population values $A(\varepsilon)$ are used as benchmarks: under the gamma model they are computed analytically as
\[
A(\varepsilon) = 1 - \frac{1}{\alpha}\left(\frac{\Gamma(\alpha+1-\varepsilon)}{\Gamma(\alpha)}\right)^{\!\frac{1}{1-\varepsilon}}, \quad \alpha + 1 - \varepsilon > 0, \; \varepsilon\neq 1,
\]
and under the inverse Gaussian model they are obtained by numerical integration. For the gamma model, the moment $\mathbb{E}[X^{1-\varepsilon}]$ requires $\varepsilon < \alpha + 1$; consequently, the curves for $\varepsilon = 5$ are absent when $\alpha \in \{2.5, 3\}$.

A key structural feature of the Atkinson family governs the expected ordering of finite-sample errors. Writing $r = 1 - \varepsilon$, the estimator $\widehat{A}(\varepsilon)$ is based on the sample power mean $\widehat{M}_r = (n^{-1}\sum X_i^r)^{1/r}$. For $r < 0$ (i.e.\ $\varepsilon > 1$), this involves negative moments of $X$, which are increasingly sensitive to small observations as $r$ decreases (equivalently, as $\varepsilon$ increases). This translates into rising finite-sample bias and MSE as $\varepsilon$ grows beyond $1$. The AHI corresponds to $\varepsilon = 2$ ($r = -1$), which lies in the moderate-sensitivity regime. Accordingly, one expects:
\begin{itemize}
  \item for $\varepsilon < 2$: $|\mathrm{Bias}(\widehat{A}(\varepsilon))| < |\mathrm{Bias}(\widehat{J})|$ and $\mathrm{MSE}(\widehat{A}(\varepsilon)) < \mathrm{MSE}(\widehat{J})$ (estimators in this range are more stable than the AHI);
  \item for $\varepsilon > 2$: $|\mathrm{Bias}(\widehat{A}(\varepsilon))| > |\mathrm{Bias}(\widehat{J})|$ and $\mathrm{MSE}(\widehat{A}(\varepsilon)) > \mathrm{MSE}(\widehat{J})$ (the AHI outperforms these more extreme specifications).
\end{itemize}

Figures~\ref{fig:ig_vs_n} and \ref{fig:gamma_vs_n} display $|\mathrm{Bias}|$ and MSE as functions of $n$, with a diverging colour scheme centred at $\varepsilon = 2$: blue dashed lines for $\varepsilon < 2$, the solid red line for $\varepsilon = 2$ (AHI), and longdashed orange/brown lines for $\varepsilon > 2$. The ordering is unambiguous across all scenarios and both distributions: the blue family lies uniformly below the AHI curve, confirming that estimators with $\varepsilon < 2$ have smaller finite-sample error than the AHI; the orange/brown family lies uniformly above it, confirming that the AHI outperforms all choices with $\varepsilon > 2$. As $n$ grows all curves converge toward zero, corroborating the consistency established in Section~\ref{large_sample}.

\begin{figure}[!ht]
\centering
\includegraphics[scale=0.52]{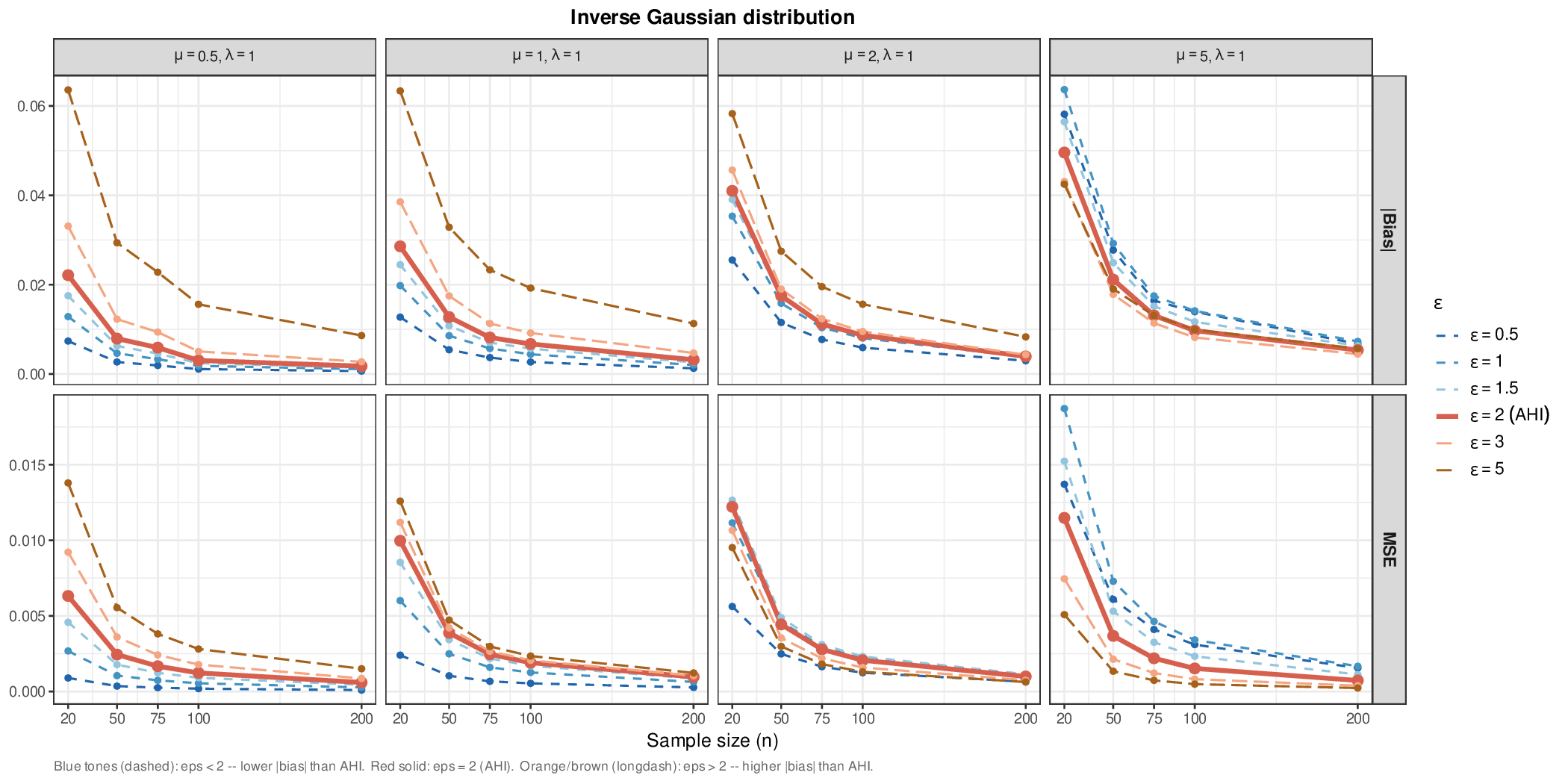}
\caption{Absolute bias and MSE of $\widehat{A}(\varepsilon)$ versus sample size $n$, under the inverse Gaussian distribution, for $\varepsilon \in \{0.5, 1, 1.5, 2, 3, 5\}$ and $(\mu,\lambda) \in \{(0.5,1),(1,1),(2,1),(5,1)\}$. Blue dashed lines ($\varepsilon < 2$): lower error than AHI; solid red line: $\varepsilon = 2$ (AHI, $\widehat{J}$); orange/brown longdashed lines ($\varepsilon > 2$): higher error than AHI.}
\label{fig:ig_vs_n}
\end{figure}

\begin{figure}[!ht]
\centering
\includegraphics[scale=0.52]{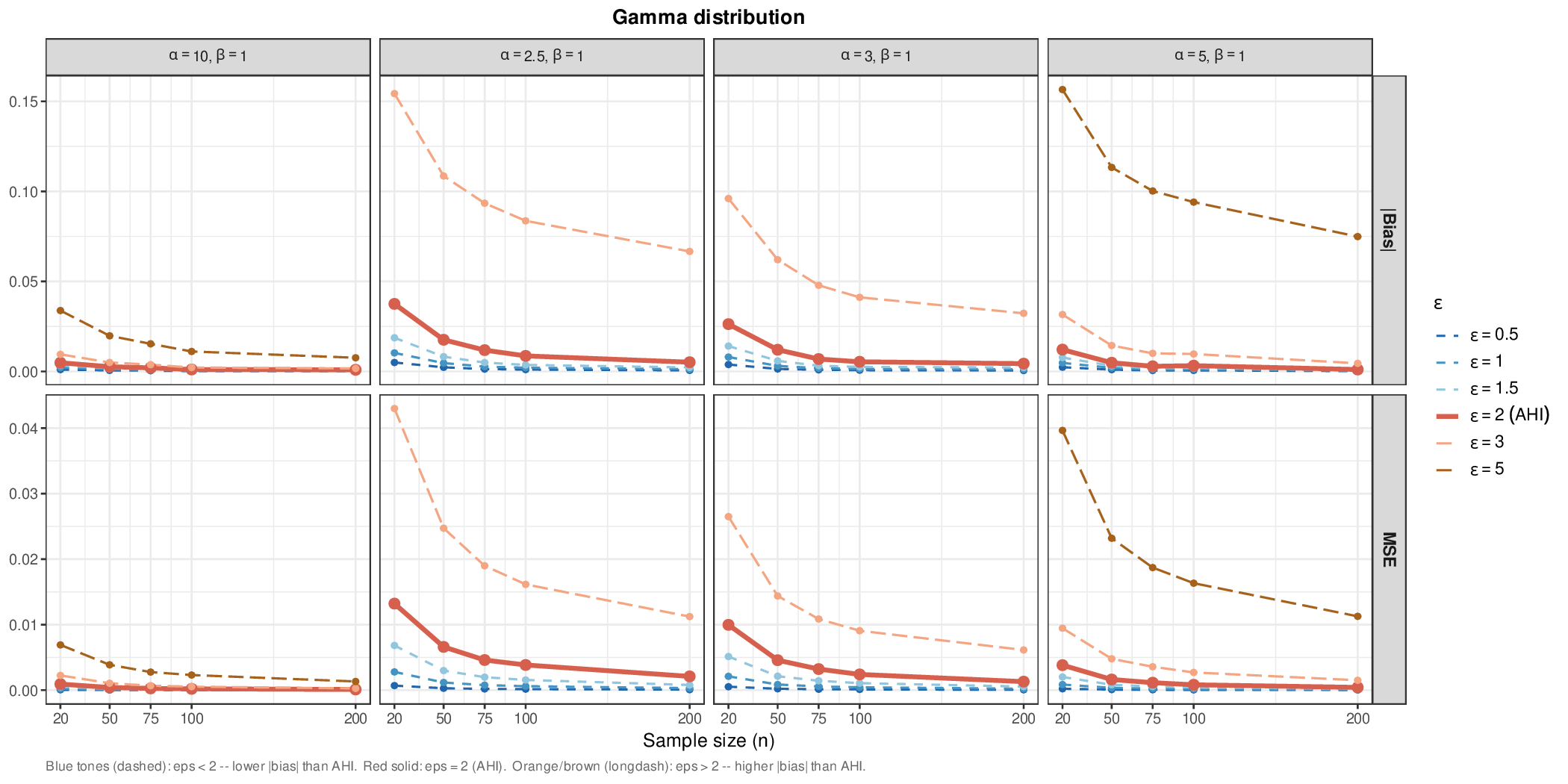}
\caption{Absolute bias and MSE of $\widehat{A}(\varepsilon)$ versus sample size $n$, under the gamma distribution, for $\varepsilon \in \{0.5, 1, 1.5, 2, 3, 5\}$ and $(\alpha,\beta) \in \{(2.5,1),(3,1),(5,1),(10,1)\}$. Blue dashed lines ($\varepsilon < 2$): lower error than AHI; solid red line: $\varepsilon = 2$ (AHI, $\widehat{J}$); orange/brown longdashed lines ($\varepsilon > 2$): higher error than AHI. For $\alpha \in \{2.5, 3\}$, the curve $\varepsilon = 5$ is absent because $\mathbb{E}[X^{1-\varepsilon}]$ does not exist.}
\label{fig:gamma_vs_n}
\end{figure}


\section{Application to real data}\label{application}

The data used in this section correspond to the Gross Domestic Product (GDP) per capita, at Purchasing Power Parity (PPP, constant 2021 international dollars), for countries in the Americas in 2023. The data were obtained from the World Bank Open Data platform \citep{WorldbankGDP2024} and, after retaining countries with positive and finite observations for the reference year, yielded a sample of $n = 34$ countries. The values are expressed in thousands of dollars to facilitate interpretation. They range from 2.96 thousand dollars (Haiti) to 74.58 thousand dollars (United States), reflecting the substantial economic heterogeneity within the continent.

Table~\ref{tab:desc-stats-americas} reports summary measures for the dataset. The mean GDP per capita is approximately 23.91 thousand dollars, whereas the median is 19.13 thousand dollars. The positive gap between mean and median, together with a skewness of approximately 1.40, confirms the right-skewed nature of the distribution, consistent with the presence of a few high-income economies at the upper tail. The large standard deviation (14.90 thousand dollars) further reinforces the substantial economic dispersion among the countries analyzed.

\begin{table}[h]
\centering
\caption{Summary statistics for the GDP per capita data (in thousands of dollars) for $n = 34$ countries in the Americas, 2023.}
\label{tab:desc-stats-americas}
\begin{tabular}{lr}
\toprule
\textbf{Measure} & \textbf{Value} \\
\midrule
$n$ & $34$ \\
Mean & $23.91$ \\
Median & $19.13$ \\
Standard deviation & $14.90$ \\
Skewness & $1.40$ \\
Minimum & $2.96$ \\
Maximum & $74.58$ \\
\bottomrule
\end{tabular}
\end{table}

Figure~\ref{fig:hist-gdp-americas} displays a histogram of the GDP per capita values. The distribution exhibits positive skewness with a heavy right tail, driven primarily by the United States and Canada, and a concentration of lower-income economies at the left.

\begin{figure}[!ht]
\centering
\includegraphics[width=0.65\textwidth]{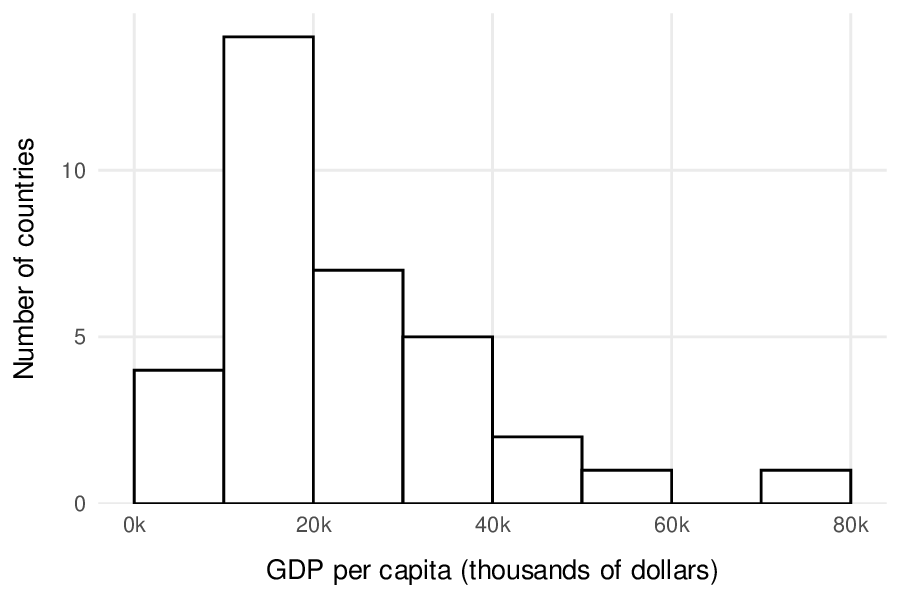}
\caption{Histogram of GDP per capita (in thousands of dollars) for $n = 34$ countries in the Americas, 2023.}
\label{fig:hist-gdp-americas}
\end{figure}

As mentioned in Section~\ref{integral_form}, since $J = A(2)$, the AHI index is a special case of the Atkinson family \citep{Atkinson1970} with inequality-aversion parameter $\varepsilon = 2$. In order to place the AHI index within the broader Atkinson family and to illustrate the role of $\varepsilon = 2$ as a natural reference point, we compute the plug-in Atkinson estimator
\[
\widehat{A}(\varepsilon)
=
1 - \frac{\displaystyle\left(\frac{1}{n}\sum_{i=1}^n X_i^{1-\varepsilon}\right)^{\!\frac{1}{1-\varepsilon}}}{\overline{X}}, \quad \varepsilon \neq 1,
\]
with the limiting formula $\widehat{A}(1) = 1 - \exp\!\bigl(n^{-1}\sum_{i=1}^n \log X_i\bigr)/\overline{X}$ as $\varepsilon \to 1$, for selected values $\varepsilon \in \{0.25, 0.5, 0.75, 1, 1.5, 2, 3, 5\}$. The plug-in estimator of the AHI index, given in \eqref{estimator-J}, is computed directly from the data as
\[
\widehat{J}
=
\widehat{A}(2)
=
1 - \frac{1}{\overline{X}\,\overline{(1/X)}}
=
0.3408,
\]
where $\overline{X} = 23.9088$ and $\overline{(1/X)} = 0.0635$ denote the sample arithmetic mean and the sample mean of the reciprocals, respectively. The corresponding harmonic mean is $H = 1/\overline{(1/X)} = 15.76$ thousand dollars, so the ratio $H/\overline{X} \approx 0.659$.

Table~\ref{tab:atkinson-comparison} reports the estimates $\widehat{A}(\varepsilon)$ for all selected values of $\varepsilon$, with the AHI row ($\varepsilon = 2$) highlighted. As established in Section~\ref{simulation_study}, the estimator $\widehat{A}(\varepsilon)$ is strictly increasing in $\varepsilon$, and this monotone ordering is clearly confirmed in the empirical results. For low values of $\varepsilon$, such as $\varepsilon = 0.25$, the index places little weight on the lower tail and yields $\widehat{A}(0.25) = 0.0422$, indicating very mild measured inequality under that criterion. As $\varepsilon$ increases, progressively greater weight is assigned to lower-income countries, and the estimated index rises steadily. At $\varepsilon = 2$, the AHI estimate $\widehat{J} = 0.3408$ captures a moderate level of inequality, in a balanced regime that penalizes lower-income economies without the extreme sensitivity associated with $\varepsilon > 2$. For $\varepsilon = 5$, the estimate reaches $\widehat{A}(5) = 0.7087$, reflecting the concentration of the distribution near the lower tail. These results confirm the interpretation of $\varepsilon = 2$ as a natural threshold within the Atkinson family, separating the regime of moderate inequality aversion from that of high aversion to low incomes.

\begin{table}[!ht]
\centering
\caption{Plug-in Atkinson estimates $\widehat{A}(\varepsilon)$ for selected values of $\varepsilon$, based on the GDP per capita data set for $n = 34$ countries in the Americas, 2023. The row $\varepsilon = 2$ corresponds to the AHI estimator $\widehat{J}$.}
\label{tab:atkinson-comparison}
\begin{tabular}{cc}
\toprule
$\varepsilon$ & $\widehat{A}(\varepsilon)$ \\
\midrule
$0.25$ & $0.0422$ \\
$0.50$ & $0.0841$ \\
$0.75$ & $0.1259$ \\
$1.00$ & $0.1679$ \\
$1.50$ & $0.2532$ \\
$\mathbf{2.00}$ & $\mathbf{0.3408}\ (= \widehat{J})$ \\
$3.00$ & $0.5099$ \\
$5.00$ & $0.7087$ \\
\bottomrule
\end{tabular}
\end{table}

\section{Concluding remarks}\label{conclusion}
It is widely known that measuring economic inequality requires indices that are at once interpretable, statistically tractable, and amenable to rigorous inference in finite samples. In this paper, we investigated the arithmetic-harmonic inequality (AHI) index as a bounded and scale-invariant measure of dispersion for positive random variables, defined through the interplay between the arithmetic and harmonic means. We derived analytical expressions for the index under the generalized inverse Gaussian family, including the inverse Gaussian and gamma distributions as important special cases. We introduced the plug-in estimator $\widehat{J}$ and established its large-sample properties, namely strong consistency, asymptotic normality, and a first-order bias expansion. The simulation results confirm that the bias and mean squared error of the estimator decrease with increasing sample size, corroborating the asymptotic theory developed in Section~\ref{large_sample}. The empirical application to GDP per capita data for $n=34$ countries in the Americas illustrates the practical utility of the index. Extensions of the AHI index to multivariate settings and mixture populations are currently under investigation.


\paragraph*{Acknowledgements}
The research was supported in part by CNPq and CAPES grants from the Brazilian government.

\paragraph*{Disclosure statement}
There are no conflicts of interest to disclose.


\bibliographystyle{apalike}


\end{document}